\documentclass{article}
 
\usepackage{amsfonts}
\usepackage{amsthm}
\usepackage{amsmath}
\usepackage{amssymb}
\usepackage{fullpage}

\newtheorem{lemma}{Lemma}[section]
\newtheorem{proposition}{Proposition}[section]
\newtheorem{theorem}{Theorem}[section]

\theoremstyle{remark}

\theoremstyle{definition}

\newcommand{\rg}{\mathrm{rg}\,}

\title{Asymptotics and analytic modes for the wave equation in similarity
coordinates}

\author{Roland Donninger\thanks{roland.donninger@univie.ac.at}\\
\small{Faculty of Physics, Gravitational Physics}\\
\small{University of Vienna}\\
\small{Boltzmanngasse 5}\\
\small{A-1090 Wien, Austria}}

\begin{document}

\maketitle

\abstract{We consider the radial wave equation in similarity coordinates within
the semigroup formalism.
It is known that the generator of the semigroup exhibits a continuum of
eigenvalues and embedded in this continuum there exists a discrete set of 
eigenvalues with analytic
eigenfunctions.
Our results show that, for sufficiently regular data, the long time behaviour of the
solution is governed by the analytic eigenfunctions. 
The same techniques are applied to the linear stability problem for the 
fundamental self--similar solution $\chi_T$ of the wave equation 
with a 
focusing power nonlinearity.
Analogous to the free wave equation, we show that the long time behaviour 
(in similarity coordinates) of linear
perturbations around $\chi_T$ is governed by analytic mode solutions.
In particular, this yields a rigorous proof for the linear stability of 
$\chi_T$ with the sharp decay rate for the perturbations.}

\section{Introduction}
\subsection{Motivation}
The focusing semilinear wave equation
\begin{equation}
\label{eq_focus} \chi_{tt}-\Delta \chi=\chi^p 
\end{equation}
for $\chi: \mathbb{R}\times \mathbb{R}^3 \to \mathbb{R}$, 
where $p>1$ is an odd integer, exhibits 
radial self--similar solutions, i.e., solutions of the form 
$\chi(t,x)=(T-t)^{-2/(p-1)}f(|x|/(T-t))$ for a function 
$f: \mathbb{R} \to \mathbb{R}$ and fixed $T>0$.
In fact, the simplest solution of this type, where $f$ is just a constant, 
can be obtained by neglecting the Laplacian in Eq.~(\ref{eq_focus})
and solving the resulting ordinary differential equation in $t$. 
We refer to this solution as the \emph{fundamental self--similar solution} and
denote it by $\chi_T$.
Although self--similar solutions do not have finite energy, one may use them 
together with smooth cut--off functions and finite speed of propagation to 
demonstrate blow up for solutions with smooth compactly supported initial data.
This observation immediately raises the question how typical such a 
self--similar blow up is.
Does it happen only for the very special initial data constructed 
by the procedure described above or can it be observed for a larger set of data?
Numerical investigations \cite{bizon} indicate that the latter is true.
Actually, there is a much stronger conjecture, namely that the 
fundamental self--similar solution describes 
the blow up behaviour for \emph{generic} large initial data.
This conjecture is based on numerical investigations for the radial equation.
In these simulations one observes 
that the future development of sufficiently large initial data converges 
to the fundamental self--similar solution near the center $r=0$ \cite{bizon}.
This indicates that $\chi_T$ has to be stable in some sense.
We remark that for $p=3$ there are also rigorous results in this direction 
(see \cite{merle1}, \cite{merle2}, \cite{merle3}).
In fact, Merle and Zaag have rigorously proved the full nonlinear stability of 
a
more general family of explicit solutions (which includes $\chi_T$) 
for the corresponding problem in one
space dimension and any $p>1$ (see \cite{merle4}, p. 48, Theorem 3).
The stability holds in the topology of the energy space.
We also mention the two recent papers \cite{merle5},
\cite{merle6} 
on interesting consequences of this result.

In order to analyse linear stability of the fundamental self--similar solution it is
convenient to introduce similarity coordinates $(\tau,\rho)$
defined by $\tau:=-\log(T-t)$ and
$\rho:=\frac{r}{T-t}$.
Since convergence is only expected near $r=0$, one requires $\rho \in (0,1)$
which corresponds to the interior of the backward lightcone of the spacetime point
$(t,r)=(T,0)$.
Transforming Eq.~(\ref{eq_focus}) to similarity coordinates, inserting the
ansatz $\chi=\chi_T+\phi$ and linearizing in $\phi$
yields a rather nasty equation of the form
\begin{equation}
\label{eq_lincss} 
\phi_{\tau \tau}+\phi_\tau+2\rho \phi_{\tau \rho}-(1-\rho^2)\phi_{\rho
\rho}-2\frac{1-\rho^2}{\rho}\phi_\rho-pc_0 \phi=0
\end{equation}
where $c_0>0$ is a constant defined by $\chi_T$. 
The first step in a heuristic stability analysis is to look for mode solutions,
i.e., one inserts the ansatz $\phi(\tau,\rho)=e^{\lambda \tau}u(\rho)$.
This yields the generalized eigenvalue problem
\begin{equation}
\label{eq_genev}
-(1-\rho^2)u''-2\frac{1-\rho^2}{\rho}u'+2\lambda \rho
u'+[\lambda(1+\lambda)-pc_0]u=0 
\end{equation}
which has two singular points at $\rho=0$ and $\rho=1$.
A necessary condition for linear stability of $\chi_T$ is the nonexistence of
mode solutions with $\mathrm{Re}\lambda>0$.
However, it is an entirely nontrivial question what kind of solutions of Eq.
(\ref{eq_genev}) one should consider as admissible.
In other words, it is not clear what boundary conditions one
should impose at the singular point $\rho=1$.
A basic Frobenius analysis shows that around $\rho=1$ 
there exists an analytic solution and a
nonanalytic one where the latter behaves as $(1-\rho)^{1-\lambda}$ for $\rho \to 1$ (we
assume noninteger $\lambda$ for simplicity).
This shows that the nonanalytic solution becomes more and more regular at
the backward lightcone as
$\mathrm{Re}\lambda$ decreases.
Hence, if $\mathrm{Re}\lambda$ is sufficiently small, 
there is no singular solution which can be excluded a priori.
Another difficulty we encounter is the fact that, 
since this is a highly non self--adjoint problem, the nonexistence of
unstable modes does not imply linear stability.

The only way to overcome these obstacles is to look for a well--posed initial
value formulation for Eq.~(\ref{eq_lincss}).
It turns out that the machinery provided by semigroup theory can be successfully
applied here.
Very sketchy, one writes Eq.~(\ref{eq_lincss}) as a first order system of the
form 
\begin{equation}
\label{eq_lincssop}
\frac{d}{d\tau}\Phi(\tau)=L \Phi(\tau) 
\end{equation}
where $L$ is a spatial differential operator which is realized as an unbounded
linear operator acting on a Banach space.
The \emph{formal} solution of this equation is $\Phi(\tau)=\exp(\tau L)\Phi(0)$ but this does not
make sense mathematically since $L$ is unbounded.
With the help of semigroup
theory one is able to construct a well--defined one--parameter family $S(\tau)$ of
operators such that the solution of Eq.~(\ref{eq_lincssop}) with initial data
$\Phi(0)$ is given by $\Phi(\tau)=S(\tau)\Phi(0)$.
Such a formulation solves the two problems described above.
First, there exists a well--defined notion of spectrum which implicitly yields
the correct boundary condition for Eq.~(\ref{eq_genev}), and, secondly, one may
use abstract results from semigroup theory to obtain growth bounds for the
solutions.

\subsection{The problem of analytic modes}
For simplicity one may first develop a semigroup formulation for the free wave
equation, i.e., Eq.~(\ref{eq_lincss}) with $c_0=0$.
This problem has recently been considered
\cite{roland1} and we have shown that there exists a semigroup $S_0(\tau)$ 
that yields the
time evolution in energy space, i.e., for very rough data.
It should be remarked that this is an interesting result per se, at least from
the mathematical point of view, since the semigroup generator is highly non
self--adjoint.
In fact, it is not even normal and its spectrum has a remarkable structure:
It
consists (essentially) of a continuum of eigenvalues filling a left 
half--plane in
the set of complex numbers.  
We review the corresponding results in Sec. \ref{sec_reviewfree}.
A special subset $\{0,-1,-2,\dots\}$ of the point spectrum consists of eigenvalues 
with analytic eigenfunctions.
From the point of view of semigroup theory there is no reason to consider these
"analytic eigenvalues" as distinguished.
However, in numerical evolutions one observes that the asymptotic behaviour (for
$\tau \to \infty$) of solutions 
is exactly described by the analytic eigenvalues and eigenfunctions \footnote{
To be precise, this is true only for data that do not have compact support
since otherwise Huygens' principle applies.}.
Therefore, the question is how to explain this behaviour.
Note that this is not a mere effect of preservation of regularity.
In the abstract approach, preservation of regularity is expressed by the
fact that domains of powers of the generator $L_0$ are invariant under the time
evolution, i.e., if $\Phi(0) \in \mathcal{D}(L_0^k)$ for $k \in \mathbb{N}$
then $S_0(\tau)\Phi(0) \in \mathcal{D}(L_0^k)$.
But one cannot get rid of "nonanalytic eigenvalues" by prescribing 
data in $\mathcal{D}(L_0^k)$ since any eigenvector of $L_0$ is by
definition also an eigenvector of $L_0^k$.
However, in Sec. \ref{sec_invh2k} we show that another class 
of higher Sobolev spaces, denoted by $\mathcal{H}^{2k}$, 
remains invariant under $S_0$.
A key observation in this respect is a certain commutator property exhibited by
the generator $L_0$, see Lemma \ref{lem_elem} below. 
The spaces $\mathcal{H}^{2k}$ are suitable to get rid of the continuum 
eigenvalues and only
analytic ones remain.
More precise, we show that initial data in $\mathcal{H}^{2k}$ can be expanded 
in a sum of the first
$2k$ analytic eigenfunctions of $L_0$ plus a remainder whose time evolution
decays faster than the rest.
This result shows in particular that the long time behaviour of solutions with
smooth initial
data is described by the analytic modes as is observed numerically.
 
\subsection{Application to the semilinear wave equation}
Numerical studies of Eq.~(\ref{eq_lincss}) exhibit a very similar behaviour as 
described above for the free wave equation: The large $\tau$ behaviour of 
linear perturbations around $\chi_T$ is precisely described by analytic modes, 
i.e., analytic
solutions of Eq.~(\ref{eq_genev}).
The techniques explained above for the free wave equation carry over to this
problem.
We obtain the analogous result (see Theorem \ref{thm_main} below) which shows
that the long time behaviour is indeed given by the analytic modes.
In particular, this result yields a rigorous proof for the linear stability of
the fundamental self--similar solution of Eq.~(\ref{eq_focus}) with the sharp
decay rate for the perturbation.

Finally, we remark that many aspects of the problem of analytic modes
are related to the 
work of N. Szpak on quasinormal
mode expansions for solutions of the wave equation \cite{szpak}.
However, the results in \cite{szpak} have been obtained by very different
methods involving the Laplace transform.
It is likely that the techniques of \cite{szpak} can also be applied to 
our problem and this would lead to a very different proof of our results.

\subsection{Notations}
To improve readability we write vectors as boldface letters and the components
are numbered by lower indices, e.g. $\mathbf{u}=(u_1,u_2)^T$.
The notation $X \hookrightarrow Y$ for two normed vector spaces $X,Y$ means that
$X$ is continuously embedded in $Y$.
When given an inner product $(\cdot|\cdot)_X$ on a vector space $X$ we denote
the induced norm by $\|\cdot\|_X$, i.e., $\|\cdot\|_X:=\sqrt{(\cdot|\cdot)_X}$.
The Cartesian product $X \times Y$ of two vector spaces $X$ and $Y$ 
with inner products
$(\cdot|\cdot)_X$ and $(\cdot|\cdot)_Y$ is implicitly assumed to be
equipped with the inner product $(\mathbf{u}|\mathbf{v})_{X \times Y}
:=(u_1|v_1)_X+(u_2|v_2)_Y$.
For a Banach space $X$ we denote by $\mathcal{B}(X)$ the space of bounded linear
operators on $X$.
For a closed operator $L: \mathcal{D}(L) \subset X \to X$ we set
$R_L(\lambda):=(\lambda -L)^{-1}$ whenever the right--hand side exists.
The resolvent set of $L$ is denoted by $\rho(L)$ and the point, continuous and
residual spectra by $\sigma_p(L)$, $\sigma_c(L)$ and $\sigma_r(L)$, respectively
(see \cite{roland1} for the precise definitions).
Finally, the expression $A \lesssim B$ means that there exists a $C>0$ such that
$A \leq CB$. 

\section{Semigroup formulation in energy space}
\label{sec_reviewfree}
In this section we review results recently obtained by the author 
\cite{roland1} on a semigroup
formulation of the free wave equation in similarity coordinates.
We define similarity coordinates $(\tau,\rho)$ as explained in the introduction
by $\tau:=-\log(T-t)$,
$\rho:=\frac{r}{T-t}$ and
consider the radial wave equation on $(3+1)$ Minkowski space,
$$ \tilde{\psi}_{tt}-\tilde{\psi}_{rr}-\frac{2}{r}\tilde{\psi}_r=0. $$
Substituting $\psi(t,r):=r\tilde{\psi}(t,r)$ yields
$$ \psi_{tt}-\psi_{rr}=0 $$
with the boundary condition $\psi(t,0)=0$ for all $t$.
We write this equation as a first order system
$$ \partial_t \Psi=\left ( \begin{array}{cc}0 & 1 \\ 1 & 0
\end{array} \right ) \partial_r \Psi $$
where $\Psi:=(\psi_t, \psi_r)^T$.
Changing to similarity coordinates we obtain
\begin{equation}
\label{eq_firstorder}
\partial_\tau \Phi=\left ( \begin{array}{cc}-\rho & 1 \\ 1 & -\rho
\end{array} \right ) \partial_\rho \Psi 
\end{equation}
where $\Phi(\tau,\rho):=\Psi(T-e^{-\tau}, \rho e^{-\tau})$.

Let $\mathcal{H}:=L^2(0,1) \times L^2(0,1)$, $\mathcal{D}(\tilde{L}_0):=\{\mathbf{u} \in
C^1[0,1] \times C^1[0,1]: u_1(0)=0\}$ and 
$$ \tilde{L}_0 \mathbf{u}(\rho):=\left( \begin{array}{c}
-\rho u_1'(\rho) +u_2'(\rho) \\ u_1'(\rho)-\rho u_2'(\rho) \end{array} \right
). $$
$\tilde{L}_0: \mathcal{D}(\tilde{L}_0) \subset \mathcal{H} \to \mathcal{H}$ is a
densely defined linear operator on the Hilbert space $\mathcal{H}$.
An operator formulation of Eq.~(\ref{eq_firstorder}) is given by
$$ \frac{d}{d\tau} \Phi(\tau)=\tilde{L}_0 \Phi(\tau) $$
for a strongly differentiable function $\Phi: [0,\infty) \to \mathcal{H}$.
We have the following result \cite{roland1}.

\begin{theorem}
\label{thm_free}
The operator $\tilde{L}_0$ is closable and its closure $L_0$ generates a
strongly continuous one--parameter semigroup $S_0: [0,\infty) \to
\mathcal{B}(\mathcal{H})$ satisfying $\|S_0(\tau)\|_{\mathcal{B}(\mathcal{H})}
\leq e^{\frac{1}{2}\tau}$ for all $\tau>0$. 

The spectrum of $L_0$ is given by 
$\sigma_p(L_0)=\{\lambda \in \mathbb{C}: \mathrm{Re}\lambda < \frac{1}{2}\}$, 
$\sigma_c(L_0)=\{\lambda \in \mathbb{C}: \mathrm{Re}\lambda = \frac{1}{2}\}$,
$\sigma_r(L_0)=\emptyset$.
\end{theorem}

\section{Semigroup formulation for more regular data}

\subsection{Invariance of higher Sobolev spaces}
\label{sec_invh2k}
We show that a certain class of higher Sobolev spaces is
invariant under the semigroup $S_0$.
For $k \in \mathbb{N}_0$ we set 
$$\mathcal{H}^{2k}:=\{\mathbf{u} \in H^{2k}(0,1)
\times H^{2k}(0,1): u_1^{(2j)}(0)=u_2^{(2j+1)}(0)=0, j \in \mathbb{N}_0, j<k\}$$
and
define an operator $D^2: \mathcal{H}^2 \to \mathcal{H}$ by
$D^2\mathbf{u}:=\mathbf{u}''$.
We have $\mathcal{H}=\mathcal{H}^0$ and equip $\mathcal{H}^{2k}$ with the inner product 
$(\mathbf{u}|\mathbf{v})_{\mathcal{H}^{2k}}
:=(\mathbf{u}|\mathbf{v})_\mathcal{H}
+(D^{2k} \mathbf{u}|D^{2k} \mathbf{v})_\mathcal{H}$.
The following lemma summarizes elementary properties.

\begin{lemma}
\label{lem_elem}
\begin{enumerate}
\item $\mathcal{H}^{2k}$ is a Hilbert space.
\item $\mathcal{H}^{2(k+1)}$ is a dense subspace of $\mathcal{H}^{2k}$ and the inclusion $\mathcal{H}^{2(k+1)} \subset \mathcal{H}^{2k}$ is
continuous.
\item The operator $D^2$ satisfies $D^2\mathcal{H}^{2(k+1)} \subset \mathcal{H}^{2k}$.
\item We have $\mathcal{H}^{2(k+1)} \subset \mathcal{D}(L_0)$ and 
$L_0\mathcal{H}^{2(k+1)} \subset \mathcal{H}^{2k}$.
\item $D^2$ and $L_0$ satisfy the commutator relation $D^2 L_0 \mathbf{u}=L_0 D^2 \mathbf{u}-
2D^2 \mathbf{u}$ for all $\mathbf{u} \in \mathcal{H}^4$.
\end{enumerate}
\end{lemma}

\begin{proof}
The proof is straightforward by inserting the definitions and using well known
properties of Sobolev spaces.
\end{proof}

As usual we define the part $L_{0,k}$ of $L_0$ in $\mathcal{H}^{2k}$ by
$\mathcal{D}(L_{0,k}):=\{\mathbf{u}\in \mathcal{D}(L_0) \cap \mathcal{H}^{2k}:
L_0 \mathbf{u} \in \mathcal{H}^{2k}\}$ and $L_{0,k}\mathbf{u}:=L_0 \mathbf{u}$. 
We show that $L_{0,k}$ generates a semigroup
on $\mathcal{H}^{2k}$.

\begin{proposition}
\label{prop_sgl0k}
The operator $L_{0,k}$ generates a strongly continuous one--parameter semigroup
$S_{0,k}: [0,\infty) \to \mathcal{B}(\mathcal{H}^{2k})$ satisfying
$\|S_{0,k}\|_{\mathcal{B}(\mathcal{H}^{2k})} \leq e^{\frac{1}{2}\tau}$.
\end{proposition}

\begin{proof}
By Lemma \ref{lem_elem} we immediately observe that $L_{0,k}$ is densely
defined since $\mathcal{H}^{2(k+1)} \subset \mathcal{D}(L_{0,k})$.

Let $(\mathbf{u}_j) \subset \mathcal{D}(L_{0,k})$ with $\mathbf{u}_j \to
\mathbf{u}$ and $L_{0,k}\mathbf{u}_j \to \mathbf{f}$ both in $\mathcal{H}^{2k}$.
Since $\mathcal{H}^{2k} \hookrightarrow \mathcal{H}$ (Lemma \ref{lem_elem}) this
implies $\mathbf{u}_j \to \mathbf{u}$, $L_0 \mathbf{u}_j \to \mathbf{f}$ in 
$\mathcal{H}$ and by the closedness
of $L_0$ we conclude $\mathbf{u} \in \mathcal{D}(L_0) \cap \mathcal{H}^{2k}$ and
$L_0 \mathbf{u}=\mathbf{f} \in \mathcal{H}^{2k}$ which
shows $\mathbf{u} \in \mathcal{D}(L_{0,k})$ and we have proved that $L_{0,k}$ is
closed.

By using the commutator relation from Lemma \ref{lem_elem} and 
integration by parts 
(cf. \cite{roland1}) we obtain
$$
\mathrm{Re}(L_{0,k}\mathbf{u}|\mathbf{u})_{\mathcal{H}^{2k}}
=\mathrm{Re}\left((L_0\mathbf{u}|\mathbf{u})_\mathcal{H}+
(L_0 D^{2k} \mathbf{u}|D^{2k} \mathbf{u})_\mathcal{H}
-2k\|D^{2k} 
\mathbf{u}\|_\mathcal{H}^2 \right ) 
\leq \frac{1}{2}
\|\mathbf{u}\|_{\mathcal{H}^{2k}}^2 
$$
for all $\mathbf{u} \in \mathcal{H}^{2(k+1)}$ and by a density argument this
estimate holds in fact for all $\mathbf{u} \in \mathcal{D}(L_{0,k})$.

Let $\mathbf{f} \in \mathcal{H}^{2k} \cap C^\infty(0,1)^2$ and define 
$F(\rho):=f_1(\rho)+\rho
f_2(\rho)+\int_0^\rho f_2(\xi)d\xi$, $u_2(\rho):=\frac{1}{1-\rho^2}\int_\rho^1
F(\xi)d\xi$ and $u_1(\rho):=\rho u_2(\rho)-\int_0^\rho f_2(\xi)d\xi$.
Then the Taylor series
expansion for $u_1$ around $\rho=0$ up to order $2k-1$ 
contains only odd powers of $\rho$ whereas the analogous series for $u_2$
up to order $2k$ contains only even powers of $\rho$.
This shows that $\mathbf{u}$ satisfies the appropriate boundary conditions at
$\rho=0$ and we conclude that $\mathbf{u} \in \mathcal{H}^{2k} \cap
\mathcal{D}(L_0)$. 
Furthermore, a direct computation yields $(1-L_0)\mathbf{u}=\mathbf{f}$ which
shows that $\mathbf{u} \in \mathcal{D}(L_{0,k})$ and
$1-L_{0,k}$ has dense range.

Invoking the Lumer--Phillips Theorem (see e.g. \cite{engel}, p. 56, Theorem
4.2.6) finishes the proof.
\end{proof}

Based on this result we are able to conclude the invariance of $\mathcal{H}^{2k}$ under the semigroup $S_0$.

\begin{lemma}
\label{lem_h2kinv}
The space $\mathcal{H}^{2k}$ is $L_0$--admissible, i.e., it is an invariant
subspace of $S_0(\tau)$, $\tau>0$, and the restriction of $S_0(\tau)$ to
$\mathcal{H}^{2k}$ is a strongly continuous semigroup on $\mathcal{H}^{2k}$
satisfying $\|S_0(\tau)\mathbf{u}\|_{\mathcal{H}^{2k}} \leq
e^{\frac{1}{2}\tau}\|\mathbf{u}\|_{\mathcal{H}^{2k}}$ for all $\mathbf{u}
\in \mathcal{H}^{2k}$ and $\tau>0$.
\end{lemma}

\begin{proof}
Let $\mathbf{f} \in \mathcal{H}^{2k}$ and $\lambda \in \rho(L_0)$.
Proposition \ref{prop_sgl0k} implies that there exists a $\mathbf{u} \in
\mathcal{D}(L_{0,k})$ such that $(\lambda-L_{0,k})\mathbf{u}=\mathbf{f}$.
However, since $L_{0,k} \subset L_0$, we have
$(\lambda-L_0)\mathbf{u}=\mathbf{f}$ and thus,
$R_{L_0}(\lambda)\mathbf{f}=\mathbf{u} \in \mathcal{H}^{2k}$.
This shows that $R_{L_0}(\lambda)\mathcal{H}^{2k} \subset \mathcal{H}^{2k}$.
By Lemma \ref{lem_elem}, the embedding $\mathcal{H}^{2k} \subset \mathcal{H}$ is 
continuous and therefore, the claim follows from Proposition \ref{prop_sgl0k}
and the theorem on admissible spaces (see e.g. \cite{pazy}, p. 123, Theorem 5.5).
\end{proof}

\subsection{Decomposition}
We improve the growth estimate 
$\|S_0(\tau)|_{\mathcal{H}^{2k}}\|_{\mathcal{B}(\mathcal{H}^{2k})} \leq
e^{\frac{1}{2} \tau}$ by a decomposition of the initial data space
$\mathcal{H}^{2k}$.
Let $\mathcal{N}$ denote the set of all $\mathbf{u}\in \mathcal{H}^{2k}$ such
that $D^{2k}\mathbf{u}=0$.
$\mathcal{N}$ is a finite dimensional (and hence closed) subspace of
$\mathcal{H}^{2k}$.

\begin{lemma}
\label{lem_s0ker}
The subspace $\mathcal{N}$ is spanned by $2k$ analytic 
functions $\mathbf{u}(\cdot, \lambda_j)$, $j=1,2,\dots,2k$ 
where each
$\mathbf{u}(\cdot, \lambda_j)$ is an eigenfunction of $L_0$ with eigenvalue 
$\lambda_j=-j+1$.
Furthermore, we have
\[ u_1^{(j-1)}(1,\lambda_j)+u_2^{(j-1)}(1,\lambda_j)\not= 0 \]
for all $j=1,2,\dots,2k$.
\end{lemma}

\begin{proof}
For $j=1,2, \dots, 2k$ set $u(\rho, \lambda_j):=(1-\rho)^{1-\lambda_j}-
(1+\rho)^{1-\lambda_j}$,
$u_1(\rho,\lambda_j):=\rho u'(\rho,\lambda_j)+(\lambda_j-1)u(\rho, \lambda_j)$ and 
$u_2(\rho, \lambda_j):=u'(\rho,\lambda_j)$.
Then $\mathbf{u}(\cdot, \lambda_j)$ is an eigenfunction of $L_0$ with 
eigenvalue $\lambda_j$ as
a straightforward computation shows.
Observe that $u(\cdot, \lambda_j)$ is an odd function 
(binomial theorem, all even powers
cancel) and therefore, $u_1(\cdot, \lambda_j)$ is odd and 
$u_2(\cdot, \lambda_j)$ is even.
This shows that $\mathbf{u}(\cdot, \lambda_j) \in \mathcal{H}^{2k}$ for each 
$j$.
Note further that $u_1(\cdot, \lambda_j)$ and $u_2(\cdot,\lambda_j)$ are 
polynomials of degree
strictly smaller than
$2k$ and thus, $D^{2k}\mathbf{u}(\cdot, \lambda_j)=0$ for each $j$ and we conclude
$\mathbf{u}(\cdot,\lambda_j) \in \mathcal{N}$.
A function in $\mathcal{H}^{2k}$ satisfies $2k$ boundary conditions
and thus, the space $\mathcal{N}$, which consists of 
$\mathbf{u} \in \mathcal{H}^{2k}$ 
with
$D^{2k}\mathbf{u}=0$, is $2k$--dimensional.
However, the $2k$ eigenfunctions $\mathbf{u}(\cdot,\lambda_j)$, 
$j=1,2,\dots,2k$
belong to $\mathcal{N}$ and they are linearly independent.
Finally, a straightforward computation yields
\[ u_1^{(j-1)}(1,\lambda_j)+u_2^{(j-1)}(1,\lambda_j)=2(-1)^j j!. \]
\end{proof}

In what follows we always assume that the eigenfunctions are normalized such
that $\|\mathbf u(\cdot,\lambda_j)\|_{\mathcal H^{2k}}=1$ and we set 
$\gamma_j:=u_1^{(j-1)}(1,\lambda_j)+u_2^{(j-1)}(1,\lambda_j)\not=0$.
For $j=1,2,\dots,2k$ we define
\[ P_j \mathbf f:=\frac{f_1^{(j-1)}(1)+f_2^{(j-1)}(1)}{\gamma_j}\mathbf{u}(\cdot,\lambda_j). \]
By Sobolev embedding it follows that $P_j$ is a bounded linear operator on $\mathcal H^{2k}$
and $P_j^2=P_j$, i.e., 
$P_j$ is a projection onto the subspace spanned by the $j$-th analytic eigenfunction.
We set $P:=\sum_{j=1}^{2k}P_j$ and remark that $P$ is also a projection
since $P_j P_\ell=0$ if $j\not= \ell$ which follows easily from the explicit form
of $\mathbf{u}(\cdot,\lambda_j)$. 
Note that the time evolution $S_0(\tau)P\mathbf{u}$
can be calculated explicitly since
\[ S_0(\tau)P\mathbf{u}=\sum_{j=1}^{2k} c_j S_0(\tau) \mathbf{u}(\cdot,\lambda_j)
=\sum_{j=1}^{2k}c_j e^{\lambda_j \tau}\mathbf{u}(\cdot,\lambda_j). \]
This suggests to decompose the space $\mathcal{H}^{2k}$ as
$\mathcal{H}^{2k}=\rg P \oplus \rg (I-P)$.
We set $\mathcal M:=\rg(I-P)$ and note that $\mathcal M\subset \mathcal H^{2k}$ is closed
since $P$ is a bounded projection on $\mathcal{H}^{2k}$.
Now we estimate the time evolution $S_0(\tau)(I-P)\mathbf{u}$.

\begin{proposition}
\label{prop_s0res}
\begin{enumerate}
\item The subspace $\mathcal{M} \subset \mathcal{H}^{2k}$ is invariant
under the semigroup $S_0$.
\item The mapping $\mathbf{u} \mapsto \|D^{2k}\mathbf{u}\|_\mathcal{H}$ defines
a norm on $\mathcal{M}$ which is equivalent to
$\|\cdot\|_{\mathcal{H}^{2k}}$.
\item For $\mathbf{u} \in \mathcal{M}$ we have the estimate
$\|S_0(\tau)\mathbf{u}\|_{\mathcal{H}^{2k}}\lesssim e^{(\frac{1}{2}-2k)\tau}
\|\mathbf{u}\|_{\mathcal{H}^{2k}}$
for all $\tau>0$.
\end{enumerate}
\end{proposition}

\begin{proof}
\begin{enumerate}
\item 
It is straightforward to check that $L_0 P_j \mathbf u=P_j L_0 \mathbf u$ for all
$\mathbf u \in \mathcal{H}^{2k}\cap \mathcal{D}(L_0)$ and $j=1,2,\dots,2k$.
Consequently, $R_{L_0}(\lambda)P\mathbf f=PR_{L_0}(\lambda)\mathbf{f}$ for all $\mathbf f\in 
\mathcal H^{2k}$ (recall that $R_{L_0}(\lambda)\mathcal H^{2k}\subset \mathcal H^{2k}$).
This implies $S_0(\tau)P \mathbf{f}=PS_0(\tau)\mathbf{f}$ for all $\mathbf f\in \mathcal H^{2k}$
and the invariance of $\mathcal M=\rg(I-P)$ follows.

\item The mapping $D^{2k}: \mathcal{H}^{2k} \to \mathcal{H}$ is linear, bounded
and surjective. 
Furthermore, $\ker D^{2k}=\rg P$ and thus, $D^{2k}|_{\mathcal {M}}: \mathcal M \to \mathcal H$
is bounded, linear, and bijective.
Since $\mathcal M$ is closed, the closed graph theorem implies
$\|D^{2k} \mathbf u\|_{\mathcal H}\gtrsim \|\mathbf u\|_{\mathcal H^{2k}}$
for all $\mathbf u \in \mathcal M$.
Trivially, we have $\|D^{2k}\mathbf{u}\|_\mathcal{H} \leq
\|\mathbf{u}\|_{\mathcal{H}^{2k}}$.

\item We equip $\mathcal{M}$ with the inner product
$(\mathbf{u}|\mathbf{v})_{\mathcal{M}}:=(D^{2k}
\mathbf{u}|D^{2k}\mathbf{v})_\mathcal{H}$. 
Then $\mathcal{M}$ is a Hilbert space and $S_0(\tau)|_{\mathcal{M}}$
defines a semigroup on $\mathcal{M}$ by assertions 1 and 2 from above. 
The generator of this semigroup is $L_{0, {\mathcal{M}}}$, the part of
$L_0$ in $\mathcal{M}$, and it satisfies
$$ \mathrm{Re}(L_{0,\mathcal{M}} \mathbf{u}|\mathbf{u})
_{\mathcal{M}} = \mathrm{Re}(L_0
D^{2k} \mathbf{u}|D^{2k}\mathbf{u})_\mathcal{H}-2k
\|\mathbf{u}\|_{\mathcal{M}}^2 \leq \left ( \frac{1}{2}-2k \right
)\|\mathbf{u}\|_{\mathcal{M}}^2 $$ 
for $\mathbf{u} \in \mathcal{D}(L_{0, \mathcal{M}})$
where we have used the commutator relation from Lemma \ref{lem_elem}
iteratively. This estimate implies
$\|S_0(\tau)\mathbf{u}\|_{\mathcal{M}}\leq
e^{(\frac{1}{2}-2k)\tau}\|\mathbf{u}\|_{\mathcal{M}}$.
However, by assertion 2 above, the norm
$\|\cdot\|_{\mathcal{M}}$ is equivalent to $\|\cdot\|_{\mathcal{H}^{2k}}$
and we arrive at the claim.
\end{enumerate}
\end{proof}

Proposition \ref{prop_s0res} implies that 
$\|S_0(\tau)(I-P)\mathbf{u}\|_{\mathcal{H}^{2k}}\lesssim 
e^{(\frac{1}{2}-2k)\tau}\|\mathbf{u}\|_{\mathcal{H}^{2k}}$ and this completes the investigation of the free wave equation in similarity coordinates. 
We end up with the result that the time evolution can be estimated as
$$ \|S_0(\tau)\mathbf{u}\|_{\mathcal{H}^{2k}}\lesssim \sum_{j=1}^{2k} c_j 
e^{\lambda_j \tau}
+e^{(\frac{1}{2}-2k)\tau}\|\mathbf{u}\|_{\mathcal{H}^{2k}} $$
for initial data $\mathbf{u} \in \mathcal{H}^{2k}$ and $\lambda_j=-j+1$.
This completely answers the question of the role of the analytic modes:
If the initial data are sufficiently regular then the long time behaviour of the solution is dominated by the first analytic eigenmodes. 
This is exactly what is observed numerically.
We also emphasize that this result implies a certain completeness property of the analytic modes:
Sufficiently regular initial data can be expanded in a sum of analytic modes plus a remainder which decays faster.
We summarize the results in a theorem.

\begin{theorem}
\label{thm_mainfree}
Let $\mathbf{u} \in \mathcal{H}^{2k}$. Then there exist constants
$c_1,\dots,c_{2k} \in \mathbb{C}$ and a function $\mathbf{g} \in
\mathcal{H}^{2k}$ such that 
$$ \mathbf{u}=\sum_{j=1}^{2k} c_j \mathbf{u}(\cdot, \lambda_j)+\mathbf{g} $$
and $\|S_0(\tau)\mathbf{g}\|_{\mathcal{H}^{2k}}\lesssim
e^{(\frac{1}{2}-2k)\tau}\|\mathbf{g}\|_{\mathcal{H}^{2k}}$ where
$\mathbf{u}(\cdot,\lambda_j)$ are normalized analytic eigenfunctions of $L_0$ with
eigenvalues $\lambda_j=-j+1$.
In particular, we have
$$ \|S_0(\tau)\mathbf{u}\|_{\mathcal{H}^{2k}}\lesssim \sum_{j=1}^{2k} 
c_j e^{\lambda_j \tau}
+e^{(\frac{1}{2}-2k)\tau}\|\mathbf{u}\|_{\mathcal{H}^{2k}} $$
for all $\tau>0$.
\end{theorem}

\section{Application to the semilinear wave equation}
We apply the previously obtained results to the linear stability problem for the fundamental
self--similar solution of Eq.~(\ref{eq_focus}).
To this end we construct a semigroup $S$ acting on $\mathcal{H}^{2k}$ that
describes the time evolution of linear perturbations of the fundamental
self--similar solution.
Througout this section we restrict ourselves to $k \in \mathbb{N}$ since
the case $k=0$ has already
been investigated in \cite{roland1}.

\subsection{Operator formulation}
Let $L' \in \mathcal{B}(\mathcal{H})$ be defined by 
$$ L' \mathbf{u}(\rho):=\left ( \begin{array}{c}pc_0 \int_0^\rho u_2(\xi)d\xi 
\\ 0 \end{array} \right ) $$
where the constant $c_0$ is given by the fundamental self--similar 
solution $\chi_T$.
Note that $L'$ leaves $\mathcal{H}^{2k}$ invariant as an odd--even argument
easily shows.
Furthermore, $D^{2k}$ and $L'$ commute, i.e.,
$D^{2k}L'\mathbf{u}=L'D^{2k}\mathbf{u}$ for all $\mathbf{u} \in
\mathcal{H}^{2k}$.

An operator formulation for the linear stability problem Eq.~(\ref{eq_lincss}) is given by
\begin{equation}
\label{eq_lincssop2}
\frac{d}{d\tau}\Phi(\tau)=L\Phi(\tau) 
\end{equation}
where $L:=L_0+L'$
and the Bounded Perturbation Theorem (see e.g. \cite{engel}) 
immediately yields the existence of a semigroup
$S: [0,\infty) \to \mathcal{B}(\mathcal{H})$ satisfying 
$\|S(\tau)\|_{\mathcal{B}(\mathcal{H})}\leq e^{(\frac{1}{2}+pc_0)\tau}$ for
all $\tau>0$ and the solution $\Phi$
of Eq.~(\ref{eq_lincssop2}) is given by $\Phi(\tau)=S(\tau)\Phi(0)$.

\subsection{Invariant subspaces, spectra} 

\begin{lemma}
\label{lem_invh2kS}
The space $\mathcal{H}^{2k}$ is $L$--admissible, i.e., it is an invariant
subspace of $S(\tau)$, $\tau>0$, and the restriction of $S(\tau)$ to
$\mathcal{H}^{2k}$ is a strongly continuous semigroup on $\mathcal{H}^{2k}$
satisfying $\|S(\tau)\mathbf{u}\|_{\mathcal{H}^{2k}} \leq
e^{(\frac{1}{2}+pc_0)\tau}\|\mathbf{u}\|_{\mathcal{H}^{2k}}$ for all $\mathbf{u}
\in \mathcal{H}^{2k}$ and $\tau>0$.
\end{lemma}

\begin{proof}
The part of $L$ in $\mathcal{H}^{2k}$ is given by 
$L_{0,k}+L'|_{\mathcal{H}^{2k}}$ since $L' \mathcal{H}^{2k} \subset
\mathcal{H}^{2k}$.
Thus, Proposition \ref{prop_sgl0k} and the Bounded Perturbation Theorem show 
that the part of $L$ in
$\mathcal{H}^{2k}$ generates a strongly continuous one--parameter semigroup
on $\mathcal{H}^{2k}$ with the same growth bound as
$S$.
Applying the same argument as in the proof of Lemma \ref{lem_h2kinv} yields the
claim.
\end{proof}

The spectrum of the free generator $L_{0,k}$ consists of a continuous part in the left half-plane
$\{z \in \mathbb{C}: \mathrm{Re} z\leq \frac12-2k\}$ and a finite number of simple eigenvalues with real parts
larger than $\frac12-2k$.
Since $L'|_{\mathcal H^{2k}}$ is compact, the spectrum
of $L_k:=L_{0,k}+L'|_{\mathcal H^{2k}}$ has the same structure, except for the fact
that the algebraic multiplicities of the isolated eigenvalues might be bigger than one (but
still finite). 
This is a consequence of the invariance of the essential spectrum under compact perturbations
and the fact that isolated eigenvalues with infinite algebraic multiplicities belong to the
essential spectrum, see \cite{kato}, p.~244, Theorem 5.35 and p.~239, Theorem 5.28.
The point spectrum of $L$ (and hence $L_k$) can be calculated by solving a 
second order ODE.
To this end we recall the definition of the operator $T(\lambda)$ from
\cite{roland1}.
The domain of $T(\lambda)$ is $\mathcal{D}(T(\lambda)):=\{u \in H^1(0,1): u \in
H^2_\mathrm{loc}(0,1), t(\lambda)u \in L^2(0,1), u(0)=0\}$ and
$T(\lambda)u:=t(\lambda)u$ where 
$$ t(\lambda)u(\rho):=-(1-\rho^2)u''(\rho)+2\lambda \rho u'(\rho)+[\lambda
(\lambda-1)-pc_0]u(\rho). $$
Then $\lambda \in \sigma_p(L)$ if and only if $\dim \ker T(\lambda)=1$
(see \cite{roland1}, Proposition 2).

Observe that $t(\lambda)u=0$ corresponds exactly to Eq.~(\ref{eq_genev}) in view
of the substitution $u(\rho) \mapsto \rho u(\rho)$.
Thus, as explained in the introduction, the semigroup approach implicitly 
yields the correct boundary condition for the generalized eigenvalue problem Eq.
(\ref{eq_genev}) that defines mode solutions.

The general solution of $t(\lambda)u=0$ can be given in terms of Legendre
functions and hence, 
the point spectrum can be calculated explicitly.
In \cite{bizon} and \cite{pohozaev} it has been shown \footnote{Note that in our
convention the whole spectrum is shifted to the right by $\frac{2}{p-1}$ 
compared to \cite{bizon}.}
that $t(\lambda)u=0$ has a nontrivial
\emph{analytic} solution if and only if $\lambda=\lambda^\pm_j$ where 
$\lambda^+_j:=1+\frac{2}{p-1}-2j$ and
$\lambda^-_j:=-\frac{2p}{p-1}-2j$ for a $j \in \mathbb{N}_0$.
We will refer to $\lambda^\pm_j$ as \emph{analytic eigenvalues}.
Moreover, it follows from \cite{bizon} that the analytic functions $u(\cdot,
\lambda^\pm_j)$ satisfying $t(\lambda^\pm_j)u(\cdot, \lambda^\pm_j)=0$ are
in fact
odd \footnote{In \cite{bizon} the corresponding analytic 
"eigenfunctions" are even polynomials of 
degree $2j$ but, according to our convention, one has to multiply them by 
$\rho$.} polynomials of degree $2j+1$. 

As before, we denote by $\mathcal{N}$ the space of all $\mathbf{u} \in
\mathcal{H}^{2k}$ such that $D^{2k} \mathbf{u}=0$.
Very similar to the free wave equation, the subspace $\mathcal{N}$ is again
spanned by analytic eigenfunctions of $L$.

\begin{lemma}
\label{lem_spanN}
The subspace $\mathcal{N}$ is spanned by $2k$ analytic 
functions $\mathbf{u}(\cdot, \lambda^\pm_j)$ for $j=0,1,2,\dots,k-1$ 
where each
$\mathbf{u}(\cdot, \lambda^\pm_j)$ is an eigenfunction of 
$L$ with eigenvalue $\lambda^\pm_j$, $\lambda^+_j=1+\frac{2}{p-1}-2j$ and
$\lambda^-_j=-\frac{2p}{p-1}-2j$.
\end{lemma}

\begin{proof}
The space $\mathcal{N}$ is $2k$--dimensional as already remarked in the proof of
Lemma \ref{lem_s0ker}.
Let $u(\cdot, \lambda^\pm_j) \not= 0$ satisfy $t(\lambda^\pm_j)u(\cdot,
\lambda^\pm_j)=0$ for $j=0,1,\dots,k-1$ and
define $\mathbf{u}(\cdot, \lambda^\pm_j)$ by 
$u_1(\rho,\lambda^\pm_j):=\rho u'(\rho,\lambda^\pm_j)+
(\lambda^\pm_j-1)u(\rho, \lambda^\pm_j)$ and
$u_2(\rho,\lambda^\pm_j):=u'(\rho,\lambda^\pm_j)$.
Then $\mathbf{u}(\cdot, \lambda^\pm_j) 
\in \mathcal{H}^{2k} \cap \mathcal{D}(L)$ since $u_1(\cdot, \lambda^\pm_j)$ is an 
odd polynomial and $u_2(\cdot, \lambda^\pm_j)$ is an even polynomial.
A straightforward computation shows $L\mathbf{u}(\cdot,
\lambda^\pm_j)=\lambda^\pm_j \mathbf{u}(\cdot, \lambda^\pm_j)$ and thus,
$\mathbf{u}(\cdot, \lambda^\pm_j)$ is an eigenfunction of $L$ with eigenvalue
$\lambda^\pm_j$.
The functions $u_1(\cdot, \lambda^\pm_j)$ and $u_2(\cdot, \lambda^\pm_j)$ are
polynomials of degree strictly smaller than $2k$ and therefore,
$D^{2k}\mathbf{u}(\cdot, \lambda^\pm_j)=0$ which shows that $\mathbf{u}(\cdot,
\lambda^\pm_j) \in \mathcal{N}$ for all $j=0,1,2,\dots,k-1$.
However, since the $\mathbf{u}(\cdot, \lambda^\pm_j)$ are eigenfunctions with
different eigenvalues they are linearly independent and they form a set of
$2k$ linearly independent functions in the $2k$--dimensional space $\mathcal{N}$.
\end{proof}

\subsection{Decomposition}

Next, we prove a crucial property which will allow us to conclude that 
all isolated analytic eigenvalues are simple.

\begin{lemma}
\label{lem_kerL}
Let $\lambda_j^\pm \in \sigma(L_k)$ be an analytic eigenvalue with eigenfunction
$\mathbf u(\cdot,\lambda_j^\pm)$ and assume
$k$ to be sufficiently large.
If $\mathbf{u} \in \mathcal D(L_k)$ satisfies $(\lambda_j^\pm - L_k)\mathbf u \in 
\ker(\lambda_j^\pm-L_k)$ then $\mathbf u=c\mathbf u(\cdot,\lambda_j^\pm)$ for a $c\in \mathbb{C}$.
\end{lemma}

\begin{proof}
If $k$ is sufficiently large we have $\ker(\lambda_j^\pm-L_k)\subset \ker D^{2(k-1)}$.
Consequently, $(\lambda_j^\pm-L_k)\mathbf{u}\in \ker(\lambda_j^\pm-L_k)$ and the commutator
relation imply
\[ 0=D^{2(k-1)}(\lambda_j^\pm-L)\mathbf{u}=(\lambda_j^\pm+2(k-1)-L)D^{2(k-1)}\mathbf{u} \]
which yields $D^{2(k-1)}\mathbf u=0$ since $\lambda_j^\pm+2(k-1) \notin \sigma(L)$ provided
$k$ is sufficiently
large.
By Lemma \ref{lem_spanN} the subspace $\ker D^{2(k-1)}$ is spanned by eigenfunctions
of $L$ and thus, $\mathbf u=c \mathbf u(\cdot,\lambda_j^\pm)$ for some $c\in \mathbb{C}$.
\end{proof}

Let $N_k^\pm \in \mathbb{N}$ be such that $\mathrm{Re}\lambda_j^\pm>\frac12-2k$ for all
$0\leq j\leq N_k^\pm$.
Then each $\lambda_j^\pm$ with $0\leq j\leq N_k^\pm$ is an isolated eigenvalue with
finite algebraic multiplicity. Consequently, there exists a spectral projection
$P_j^\pm$ of finite rank (the Riesz projection) associated to $\lambda_j^\pm$, 
which commutes with the semigroup $S(\tau)$. Moreover,
$L_k$ restricted to $\rg P_j^\pm$ is a finite-dimensional 
operator with spectrum $\{\lambda_j^\pm\}$.
By Cayley-Hamilton, the operator $(\lambda_j^\pm-L_k)|_{\rg P_j^\pm}$ is nilpotent and 
Lemma \ref{lem_kerL}
shows that
 $\rg P_j^\pm=\langle \mathbf u(\cdot,\lambda_j^\pm)\rangle$. We set
\[ P:=\sum_{j=0}^{N_k^\pm}P_j^\pm \]
and analogous to the free case we obtain a decomposition $\mathcal H^{2k}=\rg P \oplus \rg (1-P)$.

\begin{proposition}
\label{prop_SonNperp}
We have the estimate
$$\|S(\tau)\mathbf{u}\|_{\mathcal{H}^{2k}}\lesssim
e^{(\frac{1}{2}+pc_0-2k)\tau} \|\mathbf{u}\|_{\mathcal{H}^{2k}}$$
 for all $\mathbf{u} \in \mathcal M:=\rg(1-P)$
and $\tau>0$.
\end{proposition}

\begin{proof}
We define an inner product on $\mathcal M$ by
$(\mathbf{u}|\mathbf{v})_{\mathcal M}:=
(D^{2k}\mathbf{u}|D^{2k}\mathbf{v})_\mathcal{H}$.
The induced norm $\|\cdot\|_{\mathcal M}$ is equivalent to
$\|\cdot\|_{\mathcal{H}^{2k}}$ on $\mathcal M$, cf.~Proposition 
\ref{prop_s0res}, and thus, $\mathcal M$ equipped with
$(\cdot|\cdot)_{\mathcal M}$ is a Hilbert space since $\mathcal M$
is closed in $\mathcal{H}^{2k}$ by the boundedness of $P$.
The restriction $S(\tau)|_{\mathcal M}$ defines a semigroup on
$\mathcal M$ and its generator is the part of $L$ in 
$\mathcal M$, denoted by $L_{\mathcal M}$.
The generator satisfies
\begin{multline*}
\mathrm{Re}(L_{\mathcal M}\mathbf{u}|\mathbf{u})_{\mathcal M}=\\
\mathrm{Re}\left ((L_0
D^{2k}\mathbf{u}|D^{2k}\mathbf{u})_\mathcal{H}+
(L'D^{2k}\mathbf{u}|D^{2k}\mathbf{u})_\mathcal{H} \right )-2k(D^{2k}
\mathbf{u}|D^{2k}\mathbf{u})_\mathcal{H} \\
\leq \left
(\frac{1}{2}+pc_0-2k \right )\|\mathbf{u}\|_{\mathcal M}^2
\end{multline*}
for all $\mathbf{u} \in \mathcal{D}(L_{\mathcal M})$ where we have used the commutator
relation from Lemma \ref{lem_elem} and the fact that $L'$ and $D^{2k}$ commute,
as already remarked. 
This implies the growth estimate
$\|S(\tau)\mathbf{u}\|_{\mathcal M}\leq
e^{(\frac{1}{2}+pc_0-2k)\tau}\|\mathbf{u}\|_{\mathcal M}$ for all $\mathbf{u}
\in \mathcal M$ and $\tau>0$ and by the
equivalence of the norms $\|\cdot\|_{\mathcal{H}^{2k}}$ and
$\|\cdot\|_{\mathcal M}$ we arrive at the claim.
\end{proof}

Note that the growth estimate given in Proposition \ref{prop_SonNperp} is
certainly not optimal,
however, we do not make any attempts to improve it since we can make $k$
arbitrarily large.

\subsection{The time evolution of linear perturbations}
We have collected all the necessary preliminaries to conclude the result analogous
to Theorem \ref{thm_mainfree}.

\begin{theorem}
\label{thm_main}
Let $\mathbf{u} \in \mathcal{H}^{2k}$. Then there exist $2k$ constants
$c^\pm_0,\dots,c^\pm_{k-1} \in \mathbb{C}$ and a function $\mathbf{g} \in
\mathcal{H}^{2k}$ such that 
$$ \mathbf{u}=\sum_{j=0}^{k-1} \left 
(c^+_j \mathbf{u}(\cdot, \lambda^+_j)+c^-_j \mathbf{u}(\cdot,
\lambda^-_j)\right)+\mathbf{g} $$
and $\|S(\tau)\mathbf{g}\|_{\mathcal{H}^{2k}}\lesssim
e^{(\frac{1}{2}+pc_0-2k)\tau}\|\mathbf{g}\|_{\mathcal{H}^{2k}}$ where
$\mathbf{u}(\cdot,\lambda^\pm_j)$ are normalized analytic eigenfunctions of $L$ with
eigenvalues $\lambda^+_j=1+\frac{2}{p-1}-2j$ and
$\lambda^-_j=-\frac{2p}{p-1}-2j$.
In particular, we have
$$ \|S(\tau)\mathbf{u}\|_{\mathcal{H}^{2k}}\lesssim \sum_{j=0}^{k-1} \left 
(c^+_j  e^{\lambda^+_j \tau}+c^-_j e^{\lambda^-_j \tau} \right )
+e^{(\frac{1}{2}+pc_0-2k)\tau}\|\mathbf{u}\|_{\mathcal{H}^{2k}} $$
for all $\tau>0$.
\end{theorem}

By making $k$ sufficiently large we infer that the long time behaviour of smooth 
perturbations is
governed by the analytic modes and this is exactly what is observed numerically.

Furthermore, the largest eigenvalue $\lambda_0^+=1+\frac{2}{p-1}$ is known to
emerge from the time translation symmetry of the original problem (cf.
\cite{roland1}).
This apparent instability is merely an effect of the similarity coordinates and
it is therefore called the gauge instability.
Thus, for studying the question of linear stability we only allow
perturbations with $c_0^+=0$ 
(notation as in Theorem \ref{thm_main}), i.e., perturbations such that the gauge
instability is not present.
Theorem \ref{thm_main} shows that the time evolution of sufficiently regular 
perturbations $\mathbf{u} \in \mathcal{H}^{2k}$ with 
$c_0^+=0$ decays as
$\|S(\tau)\mathbf{u}\|_{\mathcal{H}^{2k}}\lesssim
e^{-\frac{p-3}{p-1}\tau}\|\mathbf{u}\|_{\mathcal{H}^{2k}}$ for $\tau \to
\infty$
and this estimate is clearly sharp.
Hence, the decay is exactly described by the largest analytic eigenvalue 
apart from the
gauge instability.
We conclude that the fundamental self--similar solution for the wave equation 
with a focusing power nonlinearity is linearly stable.

\section{Acknowledgments}
The author would like to thank Peter C. Aichelburg, Piotr Bizo\'n and Nikodem
Szpak for helpful discussions as well as the Albert Einstein Institute in Golm 
for the
invitation where the study of analytic modes has been initiated.
This work has been supported by the Austrian Fonds zur F\"orderung der
wissenschaftlichen Forschung (FWF) Project No.
P19126.

\bibliography{refasym}{}
\bibliographystyle{plain}

\end{document}